\documentclass[a4paper,12pt]{article}
\usepackage[utf8]{inputenc}
\usepackage{amsmath, amssymb, amsthm, amscd}
\usepackage{mathrsfs,mathtools}
\usepackage[text={5.8in,8.8in},centering]{geometry}
\usepackage{caption}
\usepackage{url}

\theoremstyle{plain}
\newtheorem{theorem}{Theorem}

\theoremstyle{definition}

\theoremstyle{remark}

\renewcommand{\i}{\mathrm{i}}
\newcommand{\e}{\mathrm{e}}

\newcommand{\cF}{\mathcal{F}}

\newcommand{\pder}{\partial}
\newcommand{\rI}{\mathrm{I}}

\renewcommand{\Im}{\mathrm{Im}\,}

\newcommand{\bbN}{\mathbb{N}}
\newcommand{\bbZ}{\mathbb{Z}}
\newcommand{\bbC}{\mathbb{C}}
\newcommand{\bbR}{\mathbb{R}}

\newcommand{\bF}{\mathbf{F}}
\newcommand{\bD}{\mathbf{D}}

\newcommand{\bU}{\mathbf{U}}

\begin{document}
\title{Equations of hypergeometric type\\ in the degenerate case}
\author{
  Jan Derezi\'{n}ski\footnote{The financial support of the National Science
Center, Poland, under the grant UMO-2014/15/B/ST1/00126, is gratefully
acknowledged.},
  \hskip 3ex
  Maciej Karczmarczyk\footnotemark[\value{footnote}]
\\
Department of Mathematical Methods in Physics, Faculty of Physics\\
University of Warsaw,  Pasteura 5, 02-093, Warszawa, Poland\\
email: jan.derezinski@fuw.edu.pl\\
email: maciej.karczmarczyk@fuw.edu.pl}

\maketitle

\abstract{We consider the three most important equations of hypergeometric type, 2F1, 1F1 and 1F0, in the so-called degenerate case.
In this case one of the  parameters, usually denoted $c$, is an integer and the standard basis of solutions consists of a hypergeometric-type function and a function with a logarithmic singularity. This article is devoted to a thorough analysis of the latter solution to all three equations.}

\section{Introduction}

The paper is devoted to three equations of hypergeometric type:
\begin{align}&\text{the {\em ${}_0F_1$ equation}}\notag&\\
&\big(z\partial_z^2+c\partial_z-1\big)f(z)=0,&\label{f0}\\[3ex]
&\text{the  {\em ${}_1F_1$} or the {\em confluent equation}}&\notag\\&\big(z\partial_z^2+(c-z)\partial_z-a\big)f(z)=0,&\label{f1}
\\[3ex]&\text{the  {\em ${}_2F_1$} or the {\em hypergeometric equation}}\notag&\\
&\big(z(1-z)\partial_z^2+\big(c-(a+b+1)z\big)\partial_z-ab\big)f(z)=0.&
\label{f2}\end{align}
They are probably the most important exactly solvable differential equations of mathematical physics. (Perhaps the ${}_0F_1$ equation is not so well known---by a simple transformation it is however equivalent  to the much better known Bessel equation).

The theory of these equations is quite different depending on whether the parameter $c$ is an integer or not. For integer $c$ there exist additional identities satisfied by solutions, and therefore it is more difficult to construct all solutions. One of them is analytic at $0$, but the remaining solutions have a logarithmic singularity. 
The case of integer $c$ will be called  {\em degenerate}. 

In our paper we would like to discuss systematically the equations (\ref{f0}), (\ref{f1}) and (\ref{f2}) in the degenerate case. In particular, we will introduce and analyze new special functions
$\bD(1+m;z)$,  $\bD(a;1+m;z)$, and  $\bD(a,b;1+m;z)$ useful in describing solutions of these equations in this case.

Let us remark that  the degenerate case of
(\ref{f0}), (\ref{f1}) and (\ref{f2}), even though it is in some sense  exceptional, often appears in applications. For instance, the Bessel equation with integer parameters corresponds to the degenerate case of the ${}_0F_1$ equation.

Separation of variables in the Laplacian on $\bbR^d$ leads  to the Bessel equation.
If $d$ is odd, one obtains the Bessel equation with half-integer parameters---this corresponds to the non-degenerate case. However if $d$ is even, then one obtains integer parameters---which is the degenerate case.
This is related to the fact that the resolvent of the Laplacian has a logarithmic singularity in even dimensions. When studying the wave equation one observes a similar phenomenon:
the so-called Hadamard solutions and the Feynman propagator have  a logarithmic singularity in even dimensions, e.g. in the dimension 4 of our space-time (see eg. Appendix E of \cite{H}).

In the remaining part of the introduction we will give a short resume of the results of our paper. Unlike in the rest of the paper, we will discuss in parallel the three equations  (\ref{f0}), (\ref{f1}) and (\ref{f2}).
In the rest of the paper there will be  separate sections devoted to each of these three equations. Another difference between the introduction and the rest of the paper is the choice of parameters. In the introduction we use the traditional parameters $a,b,c$. In the remaining part of the paper, instead of $a,b,c$
we will use the parameters
\begin{alignat}{5}
  \alpha&:=c-1,&&&\label{para0}\\
  \alpha&:=c-1,\quad  & \theta: =-c+2a;\label{para1}&&\\
\alpha&:=c-1,\quad&\beta: =a+b-c,\quad&\mu:=b-a.\label{para2}
\end{alignat}
These parameters,  used in \cite{D} and called there {\em Lie-algebraic}, are more convenient if we want to express symmetries of hypergeometric type equations. In the degenerate case, the parameter $\alpha$ will be usually called $m$.

\subsection{Resum\'e{} of constructions and results of the paper}

Let us start with introducing the linear operators
\begin{align}
  \mathcal{F}(c)&:=z\partial_z^2+c\partial_z-1,\label{01}\\
  \mathcal{F}(a;c)&:=z\partial_z^2+(c-z)\partial_z-a,\label{11}\\
 \mathcal{F}(a,b;c)&:=  z(1-z)\partial_z^2+\big(c-(a+b+1)z\big)\partial_z-ab
  .\label{21}\end{align}
Solving the equations  (\ref{f0}), (\ref{f1}), resp. (\ref{f2}) means
 finding the nullspace of (\ref{01}), (\ref{11}), resp. (\ref{21}).

If $c\neq0,-1,-2,\dots$, then the only solution of these equations $\sim1$ at $0$ is
\begin{align}
  \text{the {\em ${}_0F_1$ function} }&\notag\\
  F(c;z)&:=\sum_{n=0}^\infty
\frac{1}{
  (c)_n}\frac{z^j}{n!}\label{f.01}\\
&=
\e^{\mp2\sqrt{z}}F\Big(\frac{2c-1}{2};
2c-1;\pm4\sqrt{z}\Big),\notag\\
  \text{the {\em ${}_1F_1$} or the {\em confluent function} }&\notag\\
 F  (a;c;z)&:=
\sum_{n=0}^\infty
\frac{(a)_n}{
  (c)_n}\frac{z^n}{n!},\label{f.11}\\
\text{the {\em ${}_2F_1$} or the {\em hypergeometric function} }&\notag\\
  F  (a,b;c;z)&:=
\sum_{n=0}^\infty
\frac{(a)_n(b)_n}{
  (c)_n}\frac{z^n}{n!}.\label{f.21}
\end{align}

Often, it is more convenient to normalize differently
these functions:
\begin{align}
  {\bf F}  (c;z)&:=\frac{F(c;z)}{\Gamma(c)}=
\sum_{n=0}^\infty
\frac{1}{
\Gamma(c+n)}\frac{z^n}{n!},\label{f01}\\
 {\bf F}  (a;c;z)&:=\frac{F(a;c;z)}{\Gamma(c)}=
\sum_{n=0}^\infty
\frac{(a)_n}{
  \Gamma(c+n)}\frac{z^n}{n!},\label{f11}\\
 {\bf F}  (a,b;c;z)&:=\frac{F(a,b;c;z)}{\Gamma(c)}=
\sum_{n=0}^\infty
\frac{(a)_n(b)_n}{
  \Gamma(c+n)}\frac{z^n}{n!}\label{f21},
\end{align}
so that they are defined for all $c$.

It is easy to check that the equations (\ref{01}), (\ref{11}) and (\ref{21}) have another solution
\begin{align}
  z^{1-c}\mathbf{F}(2-c;z),&\label{fi01}\\
  z^{1-c}\mathbf{F}(a+1-c;2-c;z),&\label{fi11}\\
  z^{1-c}\mathbf{F}(b+1-c,a+1-c;2-c;z).\label{fi21}&
  \end{align}

However, for integer $c$,  these two solutions are proportional to one another. In fact, for $m\in\bbZ$ we have
\begin{align}\bF
(1+m;z)&=\sum_{n=\max(0,-m)}\frac{1}{n!(m+n)!}z^n,\\
\bF
(a;1+m;z)&=\sum_{n=\max(0,-m)}\frac{(a)_n}{n!(m+n)!}z^n,\\
\bF
(a,b;1+m;z)&=\sum_{n=\max(0,-m)}\frac{(a)_n(b)_n}{n!(m+n)!}z^n.
\end{align}
This easily implies the following identities  for $m\in\bbZ$:
\begin{align} \bF(1+m;z)&=z^{-m}\bF(1-m;z),\\
  (a-m)_m\bF(a;1+m;z)&=z^{-m}\bF(a-m;1-m;z),\\
  (a-m)_m(b-m)_m\bF(a,b;1+m;z)&=z^{-m}
\bF(a-m,b-m;1-m;z)
.\end{align}

Thus, if $c$ is an integer, the pairs of functions (\ref{f01}), (\ref{fi01});
(\ref{f11}), (\ref{fi11}); (\ref{f21}), (\ref{fi21}) 
do not span the whole solution space. This is the reason why this case
is called {\em degenerate}. 
As we see from the above identities, when discussing the degenerate case it is convenient to replace  $c$ with $1+m$, $m\in\bbZ$.

The pairs of functions (\ref{f01}), (\ref{fi01});
(\ref{f11}), (\ref{fi11}); (\ref{f21}), (\ref{fi21}) 
  are solutions
with a power-like behavior at $0$. The equations
 (\ref{f0}), (\ref{f1}) and (\ref{f2}) possess also other distinguished solutions. In particular, it is natural to introduce the following solutions, which have a simple behavior at infinity:
\begin{align}
U(c,z)&:=\e^{-2\sqrt z} z^{-\frac{c}{2} +\frac14}
F\Big(c-\frac{1}{2},\frac{3}{2}-c;-;-\frac{1}{4\sqrt z}\Big),\label{u0}\\
  U(a;c;z)&:=z^{-a}
  F\Big(a,a+1-c;-;-\frac{1}{z}\Big),\label{u1}\\
  \bU(a,b;c;z)&:= 	(-z)^{-a}\bF\Big(a,a-c+1;a-b+1;\frac{1}{z}\Big).\label{u2}
\end{align}
(\ref{u0}) is closely related to the MacDonald function, see (\ref{macdo}). (\ref{u1}) is sometimes called {\em Tricomi's function}.
(\ref{u2}) is one of elements of the so-called {\em Kummer's table}, which gives a list of standard solutions to the hypergeometric equation.

Both in (\ref{u0}) and (\ref{u1}) we use the  {\em ${}_2F_0$ function},
which is perhaps less known.  Note that it is not analytic at zero---it has  a branch point there. It satisfies
\begin{align}
    F  (a,b;-;z)&\sim
\sum_{n=0}^\infty
\frac{(a)_n(b)_n}{
n! }z^n\label{f.21a}
\end{align}
in  sectors $|\mathrm{arg}(z)|>\epsilon$ for any $\epsilon>0$.
For its definition and basic properties the reader can consult e.g. \cite{D}.

Now the (\ref{u0}), (\ref{u1}), resp. (\ref{u2}) are additional solutions of the equations  (\ref{f0}), (\ref{f1}), resp. (\ref{f2}) 
typically not proportional to  (\ref{f01}), 
(\ref{f11}), resp.   (\ref{f21}). They can be used in the degenerate case to obtain the full spaces of solutions.

In our paper we also analyze a different method of solving the  equations (\ref{f0}), (\ref{f1}) and (\ref{f2}) in the degenerate case.
This method is based on the observation that all solutions
not proportional to
(\ref{f01}), (\ref{f11}) and (\ref{f21}) 
have a logarithmic singularity. It is natural to look for solutions of
the equations (\ref{01}), (\ref{11}) and (\ref{21})
in the form
\begin{align}
  &\log z\bF(1+m;z)+\bD(1+m;z),\label{d01}\\
  &\log z\bF(a;1+m;z)+\bD(a;1+m;z),\label{d11}\\
  &\log(- z)\bF(a,b;1+m;z)+\bD(a,b;1+m;z).\label{d21}
\end{align}

Note that we use the so-called {\em principal branch of the logarithm}, so that the domain of $\log z$ is $\bbC\backslash]-\infty,0]$. Consequently, the domain of $\log(-z)$ is $\bbC\backslash[0,\infty[$. Solutions of the ${}_2F_1$ equation usually have a branch point at $1$, therefore it is more convenient in this case to replace $\log z$ with $\log(-z)$.

The functions 
$\bD(1+m;z)$, $\bD(a;1+m;z)$, resp.
$\bD(a,b;1+m;z)$ solve  the inhomogeneous equation
\begin{align}
  \cF(1+m)\bD(1+m;z)&=-\frac{m}{z}\bF(1+m;z)\notag\\&\hspace{4ex}-2\partial_z\bF(1+m;z),\label{inho1}\\
  \cF(a;1+m)\bD(a;1+m;z)&=\Big(1-\frac{m}{z}\Big)\bF(a;1+m;z)\notag\\&\hspace{4ex}
  -2\partial_z\bF(a;1+m;z),\label{inho2}\\
  \cF(a,b;1+m)\bD(a,b;1+m;z)&=\Big(a+b-\frac{m}{z}\Big)\bF(a,b;1+m;z)\notag\\&
  \hspace{4ex}  +2(z-1)\partial_z\bF(a,b;1+m;z).\label{inho3}
\end{align}
We will show that these equations have solutions meromorphic around $0$.
This does not fix them completely, because one can always add a multiple of
(\ref{f01}), (\ref{f11}), resp. (\ref{f21}). There exist however  canonical solutions,  which we introduce in our paper.

Using the digamma function $\psi$ (see (\ref{digamma})), for $m=0,1,2,\dots$ we define	\begin{align}
	  \bD(1+m;z) :=& \sum_{k=1}^{m} (-1)^{k-1}\frac{(k-1)!}{(m-k)!}z^{-k}\notag\\& - \sum_{k=0}^\infty \big(\psi(k+1)+\psi(k+1+m)\big)\frac{1}{k!(m+k)!}z^k\label{dd01},\\
	  \bD(a;1+m;z) :=&\sum_{k=1}^{m}(-1)^{k-1}\frac{(k-1)!(a)_{-k}}{(m-k)!}z^{-k}\notag\\& \hspace{-7ex}+\sum_{k=0}^\infty\big(\psi(a+k)-\psi(k+1)-\psi(k+1+m)\big)\frac{(a)_k}{(m+k)!k!}z^k,\label{dd11}\\
 	  \bD(a,b;1+m;z): =
&\sum_{k=1}^{m}(-1)^{k-1}\frac{(k-1)!(a)_{-k}(b)_{-k}}{(m-k)!}z^{-k}\notag\\
&\hspace{-25ex}           +\sum_{k=0}^\infty \big(
\psi(a+k)+\psi(1-b-k)-\psi(k+1)-\psi(m+1+k)\big)
\frac{(a)_k(b)_k}{(m+k)!}\frac{z^k}{k!}
.\label{dd21}
	\end{align}
We extend these definitions to negative integers by setting
\begin{align}
  \bD(1-m;z)&:=z^m\bD(1+m;z),\\
    \bD(a;1-m;z)&:=z^m\bD(a;1+m;z),\\
  \bD(a,b;1-m;z)&:=z^m\bD(a,b;1+m;z).
  \end{align}
The functions $\bD(1+m;z)$,  $\bD(a;1+m;z)$, resp.  $\bD(a,b;1+m;z)$ are solutions of
(\ref{inho1}), (\ref{inho2}), resp. (\ref{inho3}).
Therefore, (\ref{d01}), (\ref{d11}), resp. (\ref{d21}) are solutions of
 (\ref{f0}), (\ref{f1}) and (\ref{f2}).

We prove that with the definitions (\ref{dd01}), (\ref{dd11}), resp. (\ref{dd21}), the functions
(\ref{d01}), (\ref{d11}), resp. (\ref{d21}) are proportional to the special solutions (\ref{u0}), (\ref{u1}), resp. (\ref{u2}):
\begin{align}
  	U(1+m;z)& = \frac{(-1)^{m+1}}{\sqrt{\pi}}\big(\log z\cdot\bF(1+m;z) +\bD(1+m;z) \big),\\
	U(a;1+m;z)& = \frac{(-1)^{m+1}}{\Gamma(a-m)}\big(\log z\cdot \bF(a;1+m;z) + \bD(a;1+m;z)\big),\label{U-theta-m0}\\
        \bU(a,b;1+m;z)& = \frac{(-1)^{m+1}}{\Gamma(1-b)\Gamma(a-m)}\\
        &\quad\times\big(\log(- z) \cdot \bF(a,b;1+m;z) + \bD(a,b;1+m;z)\big).
\end{align}

The special functions
 $\bD(1+m;z)$,  $\bD(a;1+m;z)$, and  $\bD(a,b;1+m;z)$ satisfy various identities, which we derive in our paper.
In particular,
we compute  recurrence relations satisfied by these functions.
We show that they are very similar to the usual recurrence relations for functions $\bF(1+m;\cdot)$,
 $\bF(a;1+m;\cdot)$ and
$\bF(a,b;1+m;\cdot)$, up to terms proportional to
 $\bF(1+m;\cdot)$,
 $\bF(a;1+m;\cdot)$ and
$\bF(a,b;1+m;\cdot)$ themselves.
We also derive  quadratic relations, which involve quadratic transformations of the independent variable and doubling of parameters. 

\subsection{Bibliographical remarks}

(\ref{f0}), (\ref{f1}) and (\ref{f2}) or equivalent equations
have been studied by mathematicians for more than two centuries. Therefore, the material of our paper can be traced back to many classic papers and books, such as the textbook of Whittaker and Watson \cite{W-Wh}.

\cite{W-Wh} contains in particular a detailed analysis of the Bessel equation,
closely related to the ${}_0F_1$ equation.
In particular that the function $U(1+m;z)$ is closely related to the MacDonald function $K_m(z)$, whose degenerate case is 
analyzed in Sect. 17.71 of \cite{W-Wh}. A more complete study of the Bessel equation can be found in \cite{W}.

The  ${}_1F_1$ equation is essentially equivalent to the Whittaker equation, which is the subject of a treatise by Buchholz \cite{B}.  $U(a;1+m;z)$ is the well-known Tricomi's function---see Equation 2.25a in \cite{B} for the closely-related Whittaker function.  Buchholz analyzes its degenerate case in Sect. 2.5. He introduces a function equivalent to our $\bD(a;1+m;z)$, denoting it ${\mathfrak M}_{\kappa,\frac{m}{2}}(z)$.

Another treatise devoted to the ${}_1F_1$ equation was written by Slater
\cite{L1}. Its section 1.5 contains a discussion of the degenerate case---see in particular equation 1.5.24, equivalent to our 
\eqref{U-theta-m0}. As Slater remarks, this equation was first stated incorrectly in the literature: negative powers was missing in the formula for $U(a;1+m;z)$ in \cite{AW}. The correct formula was given 20 years later in \cite{Ar}.

The Legendre and the associated Legendre equation are the most important degenerate
cases of the ${}_2F_1$ equation. They appear e.g. in the harmonic analysis on the sphere.  They were studied e.g.  in Chap. XV of \cite{W-Wh} or in \cite{L2}. The  Legendre function of the second kind, as well as the associated Legendre function of the second kind, discussed in Sec. 15.3 of \cite{W-Wh}, are closely related to 
 $\bU(a,b;1+m;z)$.

Among the more recent references, let us mention \cite{MO}, and especially
\emph{Digital Library of Mathematical Functions}, \cite{NIST}. In 
  Equation 15.10.8 of \cite{NIST}. 
one can find the function that we call  $\bU(a,b;1+m;z)$. 
The so-called associated Legendre functions of the second kind can be found  in 15.9.16--23 of \cite{NIST}.

In our opinion,  in the literature the degenerate case of 
(\ref{f0}), (\ref{f1}) and (\ref{f2}) is usually treated
in a rather ad hoc way.
We think that this subject deserves a more systematic treartment.
To this end we   introduce the functions
$\bD(1+m;z)$,  $\bD(a;1+m;z)$, and  $\bD(a,b;1+m;z)$
and derive their various properties. Most of these properties (e.g. recurrence relations and quadratic relations) seem to be new.

\subsection{Notation}

We will often deal with multivalued analytic functions such as $z^\alpha$ and 
$\log(z)$. The standard form of these functions, called the {\em pricipal branch} has the domain $\bbC\backslash]-\infty,0]$. We can sometimes ''rotate'' these functions. For instance,  $(-z)^\alpha$ or $\log(-z)$ have
    the domain $\bbC\backslash]0,\infty]$. Note the relations
        \begin{align} z^\alpha&=\e^{\i\pi\alpha}(-z)^{\pm\alpha},\quad\pm\Im z>0;\label{multi}
          \\
          \log(-z)&=\log(z)\pm\i\pi,\quad\pm\Im z>0.
          \end{align}

	\section{The ${}_0F_1$ equation}
        \subsection{The ${}_0F_1$ function}

In this section we discuss the  ${}_0F_1$ equation, which is defined by
 the operator
\begin{eqnarray*}
\mathcal{F}_\alpha &:=&z\partial_z^2+(\alpha +1)\partial_z-1.
\end{eqnarray*}
It annihillates the  ${}_0F_1$ function $F_\alpha(z)=F(\alpha+1;z)$. We will mostly use its normalized version (\ref{f01}):
\[ {\bf F}_\alpha  (z):=\frac{F_\alpha(z)}{\Gamma(\alpha+1)}=
\sum_{n=0}^\infty
\frac{1}{
\Gamma(\alpha+1+n)}\frac{z^n}{n!}.\]

Another solution is
\[z^{-\alpha}\mathbf{F}_{-\alpha}(z).\]

\subsection{Solution with a simple behavior at infinity}

It is natural to introduce another solution
\begin{eqnarray*}
U_\alpha(z)&:=&\e^{-2\sqrt z} z^{-\frac{\alpha}{2} -\frac14}
F\Big(\frac{1}{2}+\alpha,\frac{1}{2}-\alpha;-;-\frac{1}{4\sqrt z}\Big).
\end{eqnarray*}


We have a connection formula
\begin{eqnarray}
U_\alpha (z)
&=&\frac{\sqrt\pi}{\sin\pi (-\alpha )} {\bf F}  _\alpha (z)
+\frac{\sqrt \pi}{\sin\pi \alpha }
z^{-\alpha } {\bf F}  _{-\alpha }(z),\label{connection-Um}
\end{eqnarray}
and a discrete symmetry
\[U_\alpha (z)=z^{-\alpha }U_{-\alpha }(z).\]

\subsection{Degenerate case}
If $\alpha=m\in \bbZ$, then
		\begin{equation}\label{0F1-funct}
			\bF_m(z)=\sum_{k=\max\{0,-m\}}^{\infty}\frac{1}{(k+m)!}\frac{z^k}{k!}.
		\end{equation}
                Hence
                \begin{equation}\label{degener-Fm}\bF_m(z)=z^{-m}\bF_{-m}(z),\end{equation}
so that $\bF_m$ and $z^{-m}\bF_{-m}$ are no longer linearly independent.

Assume first that    $m=0,1,2\dots$.    We look for another function annihilated  by the  operator $\cF_m$
 which has the form
	\begin{equation}\label{eq-for-D}
\log z \cdot \bF_m(z) +\bD_m(z),
	\end{equation}
	where $\bD_m(\cdot)$ is a function meromorphic around zero. Note that we have some freedom in the choice of $\bD_m(\cdot)$---we may add to it any multiple of $\bF_m(\cdot)$, i.e. the solution of the homogeneous problem.
	
	The equation
        \begin{equation}\label{0F1-op}
          \big(z\pder_z^2 +(m+1)\pder_z -1\big)
          \big(\log z \cdot \bF_m(z) +\bD_m(z)\big)=0
	\end{equation}
leads  to an inhomogeneous equation for $\bD_m(\cdot)$:
	\begin{equation}\label{eq-for-D-2}
		\big(z\pder_z^2 + (m+1)\pder_z -1\big)\bD_m(z) = -\frac{m}{z}\bF_m(z)-2\bF_m'(z).
	\end{equation}
	Suppose $\bD_m(z) = \mathop{\sum}\limits_{n=-N}^{\infty}d_n z^n$ for some $N$ (whose value will we find). Equation (\ref{eq-for-D-2}) reads then
	\begin{equation}
		\sum_{n=-N}^{\infty}d_n \big((n+m)nz^{n-1} - z^n\big) = -\sum_{n=0}^{\infty}\frac{m+2n}{(m+n)! n!}z^{n-1},
	\end{equation}
	which means that
	\begin{eqnarray}
	d_{-N}N(N-m)z^{-N-1} &+& \sum_{n=-N} ^\infty \big(d_{n+1}(n+1)(n+m+1) - d_n\big)z^n\\
	 &=& -\frac{1}{(m-1)!z} - \sum_{n=0}^\infty \frac{m+2n+2}{(m+n+1)!(n+1)!}z^n.\nonumber
	\end{eqnarray}
	For this equality to be true, the coefficients $d_n$ with negative $n$ have to fulfil
	\begin{eqnarray}
		d_{-1} &=& \frac{1}{(m-1)!},\nonumber \\
		d_{-k} &=& -(k-1)(m+1-k)d_{-k+1}  \mbox{ for }k=2,3,\dots.
	\end{eqnarray}
	This recurrence can be easily solved.  For $k=1,2,\dots$ it gives
	\begin{equation}
	d_{-k} = (-1)^{k-1}\frac{(k-1)!}{(m-k)!},
	\end{equation}
	where the factorial is understood in the sense of the $\Gamma$ function, if needed.
	
This shows us that $N=m$ (because $d_{-k}=0$ for $k>m$).
	
	For $n=0,1\dots$ we have the recursion formula
	\begin{equation}
	d_{n+1} = \frac{1}{(n+1)(n+m+1)}\Big(d_n - \frac{m+2n+2}{(m+n+1)!(n+1)!}\Big).
	\end{equation}
	
	It is solved by
	\begin{equation}\label{recur-exact}
		d_n = - \frac{1}{n!(m+n)!}\big(H_n + H_{n}(m) +{C}\big),
	\end{equation}
	where ${C}\in\bbC$ and $H_n$ and $H_n(m)$ are defined in (\ref{haka0}) and (\ref{haka}). The choice of ${C}$ corresponds to adding a multiple of $\bF_m(z)$. The formula (\ref{recur-exact}) can be proved by a simple induction argument.
	
	We define the function $\bD_m(z)$ for $m=0,1,\dots$ by
	\begin{equation}
		\bD_m(z) := \sum_{k=1}^{m} (-1)^{k-1}\frac{(k-1)!}{(m-k)!}z^{-k} - \sum_{k=0}^\infty \frac{\psi(k+1)+\psi(k+m+1)}{k!(m+k)!}z^k.
	\end{equation}
	hence we choose
        \begin{equation}{C}:=2\gamma-H_m,\end{equation} where $\gamma$ is
        Euler's constant        (see (\ref{psiha})).
	We also set
        \begin{equation}
          \bD_{-m}(z)=z^m\bD(z).
        \end{equation}
        Thus we defined $\bD_m(z)$ for all integer $m$.
For $m$ positive, it has a pole at zero of order $m$, for  $m$ negative or zero it is analytic.
	
	A close connection exists between $\bD_m$ and $U_m$ function:
        \begin{theorem} For $m\in\bbZ$,
	  \begin{equation}
	U_m(z) = \frac{(-1)^{m+1}}{\sqrt{\pi}}\big(\log z\cdot\bF_m(z) +\bD_m(z) \big).\label{close1}
	  \end{equation}
        \label{th1}  \end{theorem}

        \proof
        By (\ref{degener-Fm}) we can apply the de l'Hospital rule.
        As a preparation, we compute
        \begin{align}
          \partial_\alpha \bF_\alpha(z)&=-\sum_{j=0}^\infty\psi(\alpha+j+1)\frac{z^j}
                  {\Gamma(\alpha+j+1)j!},\\
                            \partial_\alpha \bF_\alpha(z)\Big|_{\alpha=m}&=-\sum_{j=0}^\infty\psi(m+j+1)\frac{z^j}
                                    {(m+j)!j!},\\
                                    \partial_\alpha \bF_\alpha(z)\Big|_{\alpha=-m}&=
\sum_{j=0}^{m-1}\frac{(-1)^{-m+j+1}(-m+j+1)!z^j}{j!}\\
                     & \quad              -\sum_{j=m}^\infty\psi(-m+j+1)\frac{z^j}
    {(-m+j)!j!}\\
   &=
z^m\sum_{k=1}^{m}\frac{(-1)^{k+1}(k-1)!z^{-k}}{(m-k)!}\\
                     &      \quad         -z^m\sum_{k=0}^\infty\psi(k+1)\frac{z^k}
                  {k!(k+m)!}.
        \end{align}
        Now we can write
        \begin{align}
          U_m(z)&=-\lim_{\alpha\to m}
          \frac{\sqrt\pi}{\sin\pi \alpha }\big( {\bf F}  _\alpha (z)
-
z^{-\alpha } {\bf F}  _{-\alpha }(z)\big)\\
&=\frac{(-1)^{m+1}}{\sqrt{\pi}}\Big(\partial_\alpha\bF_\alpha(z)\Big|_{\alpha=m}
+z^{-m}\partial_\alpha\bF_\alpha(z)\Big|_{\alpha=-m}+\log z\cdot \bF_m(z)\Big)\\
&=\frac{(-1)^{m+1}}{\sqrt{\pi}}
\sum_{k=1}^{m} (-1)^{k-1}\frac{(k-1)!}{(m-k)!}z^{-k}\\&\quad
- \frac{(-1)^{m+1}}{\sqrt{\pi}}\sum_{k=0}^\infty \frac{\psi(k+1)+\psi(k+m+1)}{k!(m+k)!}z^k\\
&\quad+\frac{(-1)^{m+1}}{\sqrt{\pi}}
	\log z\cdot \bF_m(z).
\end{align}
\qed

\subsection{Recurrence relations}

The function $F_\alpha$ satisfies the {\em recurrence relations}
 \begin{align*}
 \partial_z {\bf F}  _\alpha (z)&= {\bf F}  _{\alpha +1}(z),
 \\
 \left(z\partial_z+\alpha\right) {\bf F}  _\alpha (z)&= {\bf F}  _{\alpha -1}(z).
\end{align*}

The recurrence relations for $(\log z \bF_m + \bD_m)$ are the same as for $ \bF_m$.
They lead to the following recurrence relations for $\bD_m$:
		\begin{eqnarray}
			\pder_z\bD_m(z) &=& \bD_{m+1}(z) - \frac{\bF_m(z)}{z},\nonumber\\
			(z\pder_z + m)\bD_m(z) &=& \bD_{m-1}(z) - \bF_m(z).
		\end{eqnarray}
		They imply the {\em contiguity relation}
		\begin{equation}
			\bD_m(z) = \frac{1}{m}\big(\bD_{m-1}(z) - z\bD_{m+1}(z)\big).
		\end{equation}

\subsection{Bessel equation and modified Bessel equation}

	The functions $\bF_m$ and $\bD_m$ are closely related to the well-known solutions of modified Bessel equation:
	\begin{itemize}
		\item to the modified Bessel function
		\begin{equation}
			I_m(z) = \Big(\frac{z}{2}\Big)^m\bF_m\Big(\frac{z^2}{4}\Big);
		\end{equation}
		
		\item to the MacDonald function (the modified Bessel function of the second kind)
		\begin{equation}
			K_m(z) = (-1)^{m+1}\Big(\frac{z}{2}\Big)^m \Big(\log\big(\frac{z^2}{4}\big)\bF_m\big(\frac{z^2}{4}\big) +\bD_m\big(\frac{z^2}{4}\big) \Big).\label{macdo}
		\end{equation}
	\end{itemize}
	
	Similarly, the functions $\bF_m$ and $\bD_m$ are also closely related to respective solutions of Bessel equation, namely
	\begin{itemize}
		\item to the Bessel function
		\begin{equation}
		J_m(z) = \Big(\frac{z}{2}\Big)^m\bF_m\Big(-\frac{z^2}{4}\Big);
		\end{equation}
		
		\item to the Hankel functions of the first and the second type, respectively
		\begin{eqnarray}
			H^{(1)}_m(z) &=& -\frac{\i}{\pi} \Big(\e^{\frac{\i\pi}{2}}\frac{z}{2}\Big)^m \Big(\log\big(\e^{-\i\pi}\frac{z^2}{4}\big)\bF_m\big(\e^{-\i\pi}\frac{z^2}{4}\big) +\bD_m\big(\e^{-\i\pi}\frac{z^2}{4}\big) \Big),\notag\\
		H^{(2)}_m(z) &=& \frac{\i}{\pi} \Big(\e^{-\frac{\i\pi}{2}}\frac{z}{2}\Big)^m \Big(\log\big(\e^{\i\pi}\frac{z^2}{4}\big)\bF_m\big(\e^{\i\pi}\frac{z^2}{4}\big) +\bD_m\big(\e^{\i\pi}\frac{z^2}{4}\big) \Big).\label{hank}
		\end{eqnarray}
	\end{itemize}

        \section{The ${}_1F_1$ equation}

        \subsection{The  ${}_1F_1$ function}

        In the parameters $\theta,\alpha$ introduced in
(\ref{para1}),
        the  ${}_1F_1$   operator (\ref{11}) becomes
\begin{eqnarray*}
\mathcal{F}_{\theta ,\alpha}
&=&z\partial_z^2+(1+\alpha-z)\partial_z-\frac{1}{2}(1+\theta +\alpha)
,\end{eqnarray*}
It annihillates the ${}_1F_1$ function
 $F_{\theta,\alpha}(z)=F\bigl(\frac{1+\alpha+\theta }{2};\alpha+1;z\bigr)$. We will mostly use its normalized version (\ref{f11}):
\begin{align} 
 {\bf F}  _{\theta ,\alpha}(z)&:=\frac{ F  _{\theta ,\alpha}(z)}{\Gamma(\alpha+1)}
=
\sum_{n=0}^{\infty}\frac{(\frac{1+\alpha+\theta }{2})_n}{\Gamma(\alpha+n+1)n!}z^n.
		\end{align}

There is also another solution
\begin{equation}
	z^{-\alpha}\bF_{\theta,-\alpha}(z).
\end{equation}

\subsection{Tricomi's function}

One can also introduce a solution of the confluent equation with a simple behavior at infinity. It is sometimes called {\em Tricomi's function} 
\begin{align}
  U_{\theta,\alpha}(z)&:=z^{\frac{-1-\theta +\alpha}{2}}
  F\Big(\frac{1-\alpha+\theta}{2},\frac{1+\alpha+\theta}{2};-;-z^{-1}\Big).
  \end{align}

We have a connection formula
\begin{eqnarray*}
U_{\theta,\alpha}(z)
&=&\frac{\pi  {\bf F}  _{\theta ,\alpha}(z)}{\sin{\pi(-\alpha)}\Gamma\left(\frac{1+\theta -\alpha}{2}\right)}
+\frac{\pi z^{-\alpha} {\bf F}  _{\theta ,-\alpha}(z)}{\sin\pi \alpha\Gamma\left(\frac{1+\theta +\alpha}{2}\right)}
,
\end{eqnarray*}
and a discrete symmetry
\begin{equation}U_{\theta,-\alpha}(z) = z^\alpha U_{\theta,\alpha}(z).\end{equation}

	\subsection{Degenerate case}	
        If $\alpha=m\in \bbZ$, then \[\bF_{\theta,m}(z)=\sum_{n=\max\{0,-m\}}^\infty\frac{\big(\frac{\theta+m+1}{2}\big)_n}
           {n!(m+n)!} z^n.\]
           Therefore,        
           \begin{equation}\Big(\frac{\theta-m+1}{2}\Big)_m
             \bF_{\theta,m}(z)=z^{-m}\bF_{\theta,-m}(z),\end{equation}
so that $\bF_{\theta,m}$ and $z^{-m}\bF_{\theta,-m}$ are no longer linearly independent.
        
		We will look for another solution of the form
		\begin{equation}\label{eq-for-D-conf}
		\log z \cdot \bF_{\theta,m}(z) +\bD_{\theta,m}(z),
		\end{equation}
		where $\bD_{\theta,m}(\cdot)$ is  a meromorphic function around zero. Note that again we have some freedom in the choice of $\bD_{\theta,m}(\cdot)$---we may add to it any multiple of $\bF_{\theta,m}(\cdot)$.
		
 The equation
                \begin{equation}
                 \Big(z\partial_z^2+(1+\alpha-z)\partial_z-\frac{1}{2}(1+\theta +\alpha)\Big)\big(	\log z \cdot \bF_{\theta,m}(z) +\bD_{\theta,m}(z)\big)=0
                \end{equation}
                leads to an inhomogeneous equation
                \begin{equation}\Big(z\partial_z^2+(1+\alpha-z)\partial_z-\frac{1}{2}(1+\theta +\alpha)\Big)\bD_{\theta,m}(z)
                =\Big(1-\frac{m}{z}\Big)\bF_{\theta,m}(z)-2\partial_z\bF_{\theta,m}(z).
                \end{equation}

		Suppose again that $\bD_{\theta,m}=\mathop{\sum}\limits_{n=-N}^{\infty}d_n z^n$. We obtain a recurrence relation
		\begin{eqnarray}\label{recur-conf-1}
		d_{-1} &=& \frac{1}{(m-1)!(\frac{-1+m+\theta}{2})},\nonumber \\
		d_{-k} &=& -\frac{(k-1)(m+1-k)}{(\frac{1+m+\theta}{2}-k)}d_{-k+1}  \mbox{ for }k=2,3\dots,\nonumber\\
		d_{k+1}  &=& \frac{1}{(k+1)(k+m+1)}\Bigg(\Big(k+\frac{1+m+\theta}{2}\Big)d_k \notag\\&&+ \frac{(\frac{1+m+\theta}{2})_k}{(m+k+1)!(k+1)!}\Big((k+1)(m+k+1) - \Big(\frac{1+m+\theta}{2}+k\Big)(m+2k+2)\Big)\Bigg) \nonumber\\\nonumber&&\mbox{ for }k=0,1,2\dots\nonumber
		\end{eqnarray}
		
		These recursion relations
                are solved by
		\begin{eqnarray}\label{recur-conf-2}
		d_{-k} &=& (-1)^{k-1}\frac{(k-1)!\big(\frac{1+m+\theta}{2}\big)_{-k}}{(m-k)!},\mbox{ for }k=1,2\dots\\
		d_k &=& \frac{\big(\frac{1+m+\theta}{2}\big)_k}{(m+k)!k!}\Big(\psi\Big(\frac{1+m+\theta}{2}+k\Big)-H_k-H_{k}(m)+{C}\Big),\mbox{ for }k=0,1,2,\dots,m,\nonumber
		\end{eqnarray}
	        where  ${C}$ is arbitrary. To define $\bD_{\theta,m}$, we choose again \begin{equation}
                  C:=2\gamma-H_m,\end{equation} which leads to
		\begin{eqnarray}
		\bD_{\theta,m}(z) &=& \sum_{k=0}^\infty\frac{(\frac{1+m+\theta}{2})_k}{(m+k)!k!}\Big(\psi\Big(\frac{1+m+\theta}{2}+k\Big)-\psi(k+1)-\psi(k+m+1)\Big)z^k\nonumber\\&&+\sum_{k=1}^{m}(-1)^{k-1}\frac{(k-1)!\big(\frac{1+m+\theta}{2}\big)_{-k}}{(m-k)!}z^{-k}.\label{D-theta-m}
		\end{eqnarray}
                For $m=1,2,\dots$, we set
                \begin{equation}
                  \bD_{\theta,-m}(z):=z^m\bD_{\theta,m}(z).\end{equation}

The  $\bD_{\theta,m}$ function is closely related to Tricomi's function:
\begin{theorem}
  For $m\in\bbZ$, \begin{equation}
		U_{\theta,m}(z) = \frac{(-1)^{m+1}}{\Gamma(\frac{1-m+\theta}{2})}\big(\log z\cdot \bF_{\theta,m}(z) + \bD_{\theta,m}(z)\big).\label{U-theta-m}
		\end{equation}\label{th2}\end{theorem}
The proof is similar to the proof of Theorem \ref{th1}
and is omitted.

		\subsection{Recurrence relations}

		The function $\bF_{\theta,\alpha}$ fulfils the following recurrence relations:
		\begin{eqnarray}\label{recur-1F1}
			\pder_z\bF_{\theta,\alpha }(z) &=& \frac{1+\alpha +\theta}{2}\bF_{\theta+1,\alpha +1}(z), \nonumber\\
			(\pder_z-1)\bF_{\theta,\alpha }(z) &=& \frac{-1-\alpha +\theta}{2}\bF_{\theta-1,\alpha +1}(z),\nonumber\\
			(z\pder_z+\alpha -z)\bF_{\theta,\alpha }(z) &=& \bF_{\theta-1,\alpha -1}(z),\nonumber\\
			(z\pder_z+\alpha )\bF_{\theta,\alpha }(z) &=& \bF_{\theta+1,\alpha -1}(z),\notag\\
			(z\pder_z+\frac{1}{2}(1+\alpha +\theta))\bF_{\theta,\alpha }(z)&=&\frac{1+\alpha +\theta}{2}\bF_{\theta+2,\alpha }(z),\nonumber\\
			(z\pder_z+\frac{1}{2}(1+\alpha -\theta)-z)\bF_{\theta,\alpha }(z)&=&\frac{1+\alpha -\theta}{2}\bF_{\theta-2,\alpha }(z).\nonumber
		\end{eqnarray}

	The recurrence relations for $\log(z)\bF_{\theta,m} + \bD_{\theta,m}$ are the same as for $\bF_{\theta,m}$. They lead to the recurrence relations for
$ \bD_{\theta,m}$:
		\begin{eqnarray}
			\pder_z\bD_{\theta,m}(z) &=& \frac{1+\theta+m}{2}\bD_{\theta+1,m+1}(z) - \frac{\bF_{\theta,m}(z)}{z},\nonumber\\
			(\pder_z-1)\bD_{\theta,m}(z) &=& \frac{-1+\theta-m}{2}\bD_{\theta-1,m+1}(z)-\frac{\bF_{\theta,m}(z)}{z}, \nonumber\\
			(z\pder_z+m-z)\bD_{\theta,m}(z) &=& \bD_{\theta-1,m-1}(z)-\bF_{\theta,m}(z),\nonumber\\
			(z\pder_z+m)\bD_{\theta,m}(z) &=& \bD_{\theta+1,m-1}(z)-\bF_{\theta,m}(z),\notag\\
			(z\pder_z+\frac{1}{2}(1+m+\theta))\bD_{\theta,m}(z)&=&\frac{1+m+\theta}{2}\bD_{\theta+2,m}(z)-\bF_{\theta,m}(z),\nonumber\\
			(z\pder_z+\frac{1}{2}(1+m-\theta)-z)\bD_{\theta,m}(z)&=&\frac{1+m-\theta}{2}\bD_{\theta-2,m}(z)-\bF_{\theta,m}(z).\nonumber
		\end{eqnarray}
		They imply contiguous relations
		\begin{eqnarray}
			\bD_{\theta,m}(z) &=& \frac{1+m+\theta}{2}\bD_{\theta+1,m+1}(z)+\frac{1+m-\theta}{2}\bD_{\theta-1,m+1}(z),\nonumber\\
			z\bD_{\theta,m}(z) &=& \bD_{\theta+1,m-1}(z)-\bD_{\theta-1,m-1}(z),\notag\\
			(\theta+z)\bD_{\theta,m}(z) &=& \frac{1+m+\theta}{2}\bD_{\theta+2,m}(z) - \frac{1+m-\theta}{2}\bD_{\theta-2,m}(z).\nonumber
		\end{eqnarray}

                \subsection{Quadratic relations}
It is well-known that  the ${}_0F_1$ equation and the ${}_1F_1$ equation for $\theta=0$ are related by a quadratic transformation. On the level of their solutions, we have
\begin{align}
  F_\alpha(z^2)&=\e^{- 2 z} F_{0,2\alpha}( 4z),\label{double1}\\
U_\alpha(z^2)
&=2\cdot4^{\alpha}\e^{-2 z}U_{0,2\alpha}(4z).\label{double2}
\end{align}
This leads to a simple relationship between
$\bD_m$ and $\bD_{0,2m}$:
\begin{theorem}
  \begin{equation}\bD_m(z^2)=\frac{2(-4)^m\sqrt{\pi}}{\Gamma(\frac{1}{2}-m)}\e^{-2z}\big(\log(4)\bF_{0,2m}(4z)+\bD_{0,2m}(4z)\big).
\label{double5}\end{equation}\end{theorem}

\proof
Inserting (\ref{double1}) into (\ref{close1}) we obtain
\begin{equation}
  U_m(z^2)=\frac{(-1)^{m+1}}{\sqrt{\pi}}\Big(2\log (z)\frac{\Gamma(1+2m)}{\Gamma(1+m)}\e^{-2z}\bF_{0,2m}(4z)+\bD_m(z^2)\Big).
\label{double3}\end{equation}
Inserting (\ref{double1}) into (\ref{U-theta-m}) we obtain
\begin{align}&2\cdot4^m\e^{-2z}U_{0,2m}(z)\notag\\
=&-\frac{2\cdot4^m}{\Gamma(\frac12-m)}\e^{-2z}
\Big(  \big((\log 4    +\log(z)\big)\bF_{0,2m}(4z)+
  \bD_{0,2m}(4z)\Big).
  \label{double4}\end{align}
Now by (\ref{double2})
we have (\ref{double3})=(\ref{double4}).
Using the identity
\begin{equation}
  \frac{(-1)^m\Gamma(1+2m)}{\sqrt{\pi}\Gamma(1+m)}=\frac{4^m}{\Gamma(\frac12-m)},
  \label{gamma}\end{equation}
we see that the terms with $\log(z)$ cancel and we obtain (\ref{double5}).
\qed

		\section{The ${}_2F_1$ equation}
		\subsection{The ${}_2F_1$ function}
                In the parameters $\alpha,\beta,\mu$ introduced in (\ref{para2}), 
                the  ${}_2F_1$   operator (\ref{21}) becomes
			\begin{align}
			\mathcal{F}_{\alpha,\beta,\mu}:=& z(1-z)\pder_z^2 + \big((1+\alpha)(1-z)-(1+\beta)z\big)\pder_z\notag\\
			&+ \frac{\mu^2}{4}-\frac{1}{4}(\alpha+\beta+1)^2.
			\end{align}
                        It annihillates the ${}_2F_1$ function
$ F_{\alpha,\beta ,\mu }(z)=F\bigl(
\frac{1+\alpha+\beta -\mu}{2},\frac{1+\alpha+\beta +\mu}{2};1+\alpha;z\bigr)$. We will mostly use its normalized version (\ref{f21}):
			\begin{equation}
			\bF_{\alpha,\beta,\mu} (z):=\frac{      			F_{\alpha,\beta,\mu} (z)}{\Gamma(\alpha+1)}=                  
	                \sum_{n=0}^\infty\frac{\big(\frac{1+\alpha+\beta-\mu}{2}\big)_n\big(\frac{1+\alpha+\beta+\mu}{2}\big)_n}{\Gamma(1+\alpha+n)n!}z^n.
			\end{equation}

			There is also another solution with a power-like behavior at zero:
			\begin{equation}
				z^{-\alpha}\bF_{-\alpha,\beta,-\mu}(z).
			\end{equation}

		\subsection{Solution with a simple behaviour at infinity}

The following function is annihilated by $\cF_{\alpha,\beta,\mu}$ and behaves as $\frac{1}{\Gamma(1-\mu)}(-z)^{\frac{-1-\alpha-\beta+\mu}{2}}$ at $\infty$ (see e.g. \cite{D}):
 \begin{equation}
   \bU_{\alpha,\beta,\mu}(z):= 	(-z)^{\frac{-1-\alpha-\beta+\mu}{2}}\bF_{-\mu,\beta,-{\alpha}}(z^{-1}).\end{equation}
		It can be expressed with use of $\bF_{\alpha,\beta,\mu}$ function:
		\begin{align}\notag
			\bU_{\alpha,\beta,\mu}(z)& = -\frac{\pi}{\sin(\pi \alpha)}\bigg(\frac{\bF_{\alpha,\beta,\mu}(z)}{\Gamma(\frac{1-\alpha-\beta-\mu}{2})\Gamma(\frac{1-\alpha+\beta-\mu}{2})}\\&\hspace{5ex}- \frac{(-z)^{-\alpha}\bF_{-\alpha,\beta,-\mu}(z)}{\Gamma(\frac{1+\alpha+\beta-\mu}{2})\Gamma(\frac{1+\alpha-\beta-\mu}{2})}\bigg).
		\end{align}

                We have a set of identities
                \begin{align}
                  &\bU_{\alpha,\beta,\mu}(z)\\
                  =\ & (-z)^{-\alpha}\bU_{-\alpha,\beta,\mu}(z)\\
                                    =\ &(1-z)^{-\beta}\bU_{\alpha,-\beta,\mu}(z)\label{also}\\
                          =\ &(-z)^{-\alpha}(1-z)^{-\beta}\bU_{-\alpha,-\beta,\mu}(z),
                \end{align}
                which are essentially a part of the so-called Kummer table, see e.g. \cite{D}. They follow by the following argument: all of them are annihilated by $\cF_{\alpha,\beta,\mu}$ and behave like
$\frac{1}{\Gamma(1-\mu)}(-z)^{\frac{-1-\alpha-\beta+\mu}{2}}$ at $\infty$. These conditions determine uniquely a solution to the hypergeometric equation.
                
                We have another identity
                \begin{equation}
                \bU_{\alpha,\beta,\mu}(z)=\e^{\mp\i\frac{\pi}{2}(-1-\alpha-\beta+\mu)}
                \bU_{\beta,\alpha,\mu}(1-z),\quad \pm\Im z>0.
                \end{equation}
                Indeed, 
$  \bU_{\beta,\alpha,\mu}(1-z)$ is annihilated by $\cF_{\alpha,\beta,\mu}$ and behaves as $\frac{1}{\Gamma(1-\mu)}z^{\frac{-1-\alpha-\beta+\mu}{2}}$ at $\infty$. Then we use (\ref{multi}).

		\subsection{Degenerate case}
		
		If $\alpha=m\in \bbZ$, then \[\bF_{m,\beta,\mu}(z)=\sum_{n=\max\{0,-m\}}^\infty\frac{\big(\frac{1+m+\beta-\mu}{2}\big)_n\big(\frac{1+m+\beta+\mu}{2}\big)_n}
		{n!(m+n)!} z^n.\]
		Therefore,        
		\begin{align}
z^{-m}\bF_{-m,\beta,-\mu}(z)&=	\Big(\frac{1-m\pm \beta-\mu}{2}\Big)_m\Big(\frac{1-m\pm\beta+\mu}{2}\Big)_m\bF_{m,\beta,\mu}(z)\label{degenerate-Fabm2}
\\&=
(-1)^m\Big(\frac{1-m+\beta\pm\mu}{2}\Big)_m\Big(\frac{1-m-\beta\pm\mu}{2}\Big)_m\bF_{m,\beta,\mu}(z).\label{degenerate-Fabm2a}
                \end{align}
		Hence $\bF_{m,\beta,\mu}$ and $z^{-m}\bF_{-m,\beta,-\mu}$ are no longer linearly independent.

Note that
(\ref{degenerate-Fabm2}) and  (\ref{degenerate-Fabm2a}) contain 4 ways of writing
                 the coefficient in front of $\bF_{m,\beta,\mu}$---this follows from the identity (\ref{pochhammer}).

 		For $m\in\bbN$ we will look for another solution of the form 
		\begin{equation}\label{eq-for-D-hyper}
\log(- z)\cdot\bF_{m,\beta,\mu}(z) + \bD_{m,\beta,\mu}(z).
		\end{equation}
		Again, this does not fix $\bD_{m,\beta,\mu}$---we may add to it any multiple of $\bF_{m,\beta,\mu}$. Inserting (\ref{eq-for-D-hyper}) into the hypergeometric equation yields the recurrence relations
		\begin{eqnarray}
		d_{-1} &=& \frac{1}{(m-1)!}\,\Big(\frac{2}{m-1+\beta-\mu}\Big)\Big(\frac{2}{m-1+\beta+\mu}\Big),\notag\\
		d_{-k} &=& -\frac{(k-1)(m+1-k)}{\Big(\frac{1+m+\beta-\mu}{2}-k\Big)\Big(\frac{1+m+\beta+\mu}{2}-k\Big) }d_{-k+1},\nonumber\\
		d_{k+1} &=& \frac{1}{(k+1)(k+1+m)}\Bigg(\Big(\frac{1+m+\beta-\mu}{2}+k\Big)\Big(\frac{1+m+\beta+\mu}{2}+k\Big) d_{k}\nonumber\\ 
				&& + \frac{\big(\frac{1+m+\beta-\mu}{2}\big)_k\big(\frac{1+m+\beta+\mu}{2}\big)_k}{k!(m+k)!}\bigg((1+\beta+m+2k) \nonumber\\
				&&\hspace{0.75cm}-\frac{\Big(\frac{1+m+\beta-\mu}{2}+k\Big)\Big(\frac{1+m+\beta+\mu}{2}+k\Big)}{(1+m+k)(k+1)}(2k+m+2)\bigg)\Bigg).\nonumber
		\end{eqnarray}
		These recursion relations
                are solved by
		\begin{eqnarray}\label{recur-hyper-2}
		d_{-k} &=& (-1)^{k-1}\frac{(k-1)!(\frac{1+m+\beta+\mu}{2})_{-k}
             (\frac{1+m+\beta-\mu}{2})_{-k}}{(m-k)!},\mbox{ for }k=1,2\dots,m,\\
		d_k &=&               \bigg(H_k\Big(\frac{1+m+\beta+\mu}{2}\Big)+H_k\Big(\frac{1+m+\beta-\mu}{2}\Big)\nonumber \\&&\hspace{0.75cm} -H_k- H_{k}(m)+c\bigg)\frac{(\frac{1+m+\beta+\mu}{2})_k(\frac{1+m+\beta-\mu}{2})_k}{(m+k)!k!},\mbox{ for }k=0,1,2,\dots,\nonumber
		\end{eqnarray}
	where  ${C}$ is arbitrary. 
We introduce a particular solution of these relations:
		\begin{eqnarray}
		\bD_{m,\beta,\mu}(z) &=& \sum_{k=0}^\infty \bigg(\psi\Big(\frac{1+m+\beta-\mu}{2}+k\Big)+\psi\Big(\frac{1-m-\beta-\mu}{2}-k\Big)\nonumber \\&&\hspace{0.75cm} -\psi(k+1)- \psi(m+k+1)\bigg)\frac{(\frac{1+m+\beta+\mu}{2})_k(\frac{1+m+\beta-\mu}{2})_k}{(m+k)!k!}z^k\nonumber\\
		&&+\sum_{k=1}^{m}(-1)^{k-1}\frac{(k-1)!(\frac{1+m+\beta+\mu}{2})_{-k}
             (\frac{1+m+\beta-\mu}{2})_{-k}}{(m-k)!}z^{-k}.
		\end{eqnarray}
                For $m=1,2,\dots$, we set
                \begin{equation}
                  \bD_{-m,\beta,\mu}(z):=z^m\bD_{m,\beta,\mu}(z).\end{equation}

 The function $\bU_{m,\beta,\mu}$ is closely related to $\bD_{m,\beta,\mu}$:
 \begin{theorem} For $m\in\bbZ$,
		\begin{equation}\label{U-def}
	\bU_{m,\beta,\mu}(z) = \frac{(-1)^{m+1}}{\Gamma(\frac{1-m-\beta-\mu}{2})\Gamma(\frac{1-m+\beta-\mu}{2})}\big(\log(- z) \cdot \bF_{m,\beta,\mu}(z) + \bD_{m,\beta,\mu}(z)\big).
		\end{equation}\end{theorem}

 \proof Note that  the minus case of (\ref{degenerate-Fabm2a}) can be rewritten as
		\begin{equation}	\frac{(-z)^{-m}\bF_{-m,\beta,-\mu}(z)}{\Gamma(\frac{1+m+\beta-\mu}{2})\Gamma(\frac{1+m-\beta-\mu}{2})} = \frac{\bF_{m,\beta,\mu}(z)}{\Gamma(\frac{1-m-\beta-\mu}{2})\Gamma(\frac{1-m+\beta-\mu}{2})}.\label{degenerate-Fabm}
		\end{equation}
                Therefore, we can apply
                the de l'Hospital rule. As a preparation for this we compute
                \begin{align}            &    \partial_\alpha\frac{
    \bF_{\alpha,\beta,\mu}(z)}
  {\Gamma(\frac{1-\alpha-\beta-\mu}{2})
    \Gamma(\frac{1-\alpha+\beta-\mu}{2})}\\
    =&\sum_{k=0}^\infty\Bigg(\frac12\psi\Big(\frac{1-\alpha-\beta-\mu}{2}\Big)
    +\frac12\psi\Big(\frac{1-\alpha+\beta-\mu}{2}\Big)\notag\\
    &\quad
    +\frac12H_k\Big(\frac{1+\alpha+\beta-\mu}{2}\Big)
    +\frac12H_k\Big(\frac{1+\alpha+\beta+\mu}{2}\Big)
    -\psi(1+\alpha+k)\Bigg)\notag
    \\
    &\times\frac{\big(\frac{1+\alpha+\beta-\mu}{2}\big)_k
      \big(\frac{1+\alpha+\beta+\mu}{2}\big)_k z^k }
      {\Gamma\big(\frac{1-\alpha-\beta-\mu}{2}\big)
        \Gamma\big(\frac{1-\alpha+\beta-\mu}{2}\big)
        \Gamma(1+\alpha+k)k!}.\notag
                \end{align}
                Thus,
                \begin{align}
   &                \partial_\alpha
    \frac{\bF_{\alpha,\beta,\mu}(z)}
         {\Gamma(\frac{1-\alpha-\beta-\mu}{2})\Gamma(\frac{1-\alpha+\beta-\mu}{2})}\Big|_{\alpha=m}\\
         =&\sum_{k=0}^\infty\Bigg(\frac12\psi\Big(\frac{1-m-\beta-\mu}{2}\Big)
    +\frac12\psi\Big(\frac{1-m+\beta-\mu}{2}\Big)\notag\\
    &\quad
    +\frac12H_k\Big(\frac{1+m+\beta-\mu}{2}\Big)
    +\frac12H_k\Big(\frac{1+m+\beta+\mu}{2}\Big)
    -\psi(1+m+k)\Bigg)\notag
    \\
    &\times\frac{\big(\frac{1+m+\beta-\mu}{2}\big)_k
      \big(\frac{1+m+\beta+\mu}{2}\big)_k z^k }
      {\Gamma\big(\frac{1-m-\beta-\mu}{2}\big)
        \Gamma\big(\frac{1-m+\beta-\mu}{2}\big)
        (m+k)!k!}\notag,\\
         &\partial_\alpha
    \frac{\bF_{\alpha,\beta,\mu}(z)}
         {\Gamma(\frac{1-\alpha-\beta-\mu}{2})\Gamma(\frac{1-\alpha+\beta-\mu}{2})}\Big|_{\alpha=-m}\\=&         
\sum_{k=0}^{m-1}\frac{(-1)^{1-m+k}(-1+m-k)!\big(\frac{1-m+\beta-\mu}{2}\big)_k
      \big(\frac{1-m+\beta+\mu}{2}\big)_k z^k }
      {\Gamma\big(\frac{1+m-\beta-\mu}{2}\big)
        \Gamma\big(\frac{1+m+\beta-\mu}{2}\big)
        k!}\notag
      \\&+         
\sum_{k=m}^\infty\Bigg(\frac12\psi\Big(\frac{1+m-\beta-\mu}{2}\Big)
+\frac12\psi\Big(\frac{1+m+\beta-\mu}{2}\Big)\notag\\
    &\quad
    +\frac12H_k\Big(\frac{1-m+\beta-\mu}{2}\Big)
    +\frac12H_k\Big(\frac{1-m+\beta+\mu}{2}\Big)
    -\psi(1-m+k)\Bigg)\notag
    \\
    &\times\frac{\big(\frac{1-m+\beta-\mu}{2}\big)_k
      \big(\frac{1-m+\beta+\mu}{2}\big)_k z^k }
      {\Gamma\big(\frac{1+m-\beta-\mu}{2}\big)
        \Gamma\big(\frac{1+m+\beta-\mu}{2}\big)
        (-m+k)!k!}\notag
      \\=&         
\sum_{k=0}^{m-1}\frac{(-1)^{1+k}(-1+m-k)!\big(\frac{1+m+\beta-\mu}{2}\big)_{k-m}
      \big(\frac{1+m+\beta+\mu}{2}\big)_{k-m} z^k }
      {\Gamma\big(\frac{1-m-\beta-\mu}{2}\big)
        \Gamma\big(\frac{1-m+\beta-\mu}{2}\big)
        k!}\notag
      \\&+         
(-z)^m\sum_{k=0}^\infty\Bigg(\frac12\psi\Big(\frac{1+m-\beta-\mu}{2}\Big)
 +\frac12\psi\Big(\frac{1+m+\beta-\mu}{2}\Big)\notag\\
    &\quad
    +\frac12H_{k+m}\Big(\frac{1-m+\beta-\mu}{2}\Big)
    +\frac12H_{k+m}\Big(\frac{1-m+\beta+\mu}{2}\Big)
    -\psi(1+k)\Bigg)\notag
    \\
    &\times\frac{\big(\frac{1+m+\beta-\mu}{2}\big)_k
      \big(\frac{1+m+\beta+\mu}{2}\big)_k z^k }
      {\Gamma\big(\frac{1-m+\beta-\mu}{2}\big)
        \Gamma\big(\frac{1-m-\beta-\mu}{2}\big)
        k!(k+m)!},\label{shifted}   
                \end{align}
               where we shifted the variable $k$ by $m$ and used a few identities for the Pochhammer symbol in (\ref{shifted}).
               Now, recalling that $F_{-\alpha,\beta,-\mu}(z)=F_{-\alpha,\beta,\mu}(z)$,
               we can write
\begin{align}\notag
  \bU_{m,\beta,\mu}(z)&=
    -\lim_{\alpha\to m}\frac{\pi}{\sin(\pi \alpha)} \bigg(\frac{\bF_{\alpha,\beta,\mu}(z)}{\Gamma(\frac{1-\alpha-\beta-\mu}{2})\Gamma(\frac{1-\alpha+\beta-\mu}{2})}- \frac{(-z)^{-\alpha}\bF_{-\alpha,\beta,-\mu}(z)}{\Gamma(\frac{1+\alpha+\beta-\mu}{2})\Gamma(\frac{1+\alpha-\beta-\mu}{2})}\bigg)
\\
=&(-1)^{m+1}\Bigg(   \partial_\alpha
    \frac{\bF_{\alpha,\beta,\mu}(z)} {\Gamma(\frac{1-\alpha-\beta-\mu}{2})\Gamma(\frac{1-\alpha+\beta-\mu}{2})}\Big|_{\alpha=m}\notag\\
&+ (-z)^{-m}         \partial_\alpha
    \frac{\bF_{\alpha,\beta,\mu}(z)} {\Gamma(\frac{1-\alpha-\beta-\mu}{2})\Gamma(\frac{1-\alpha+\beta-\mu}{2})}\Big|_{\alpha=-m}\notag\\&
+\log(-z)(-z)^{-m}    \frac{\bF_{-m,\beta,\mu}(z)}
{\Gamma(\frac{1+m-\beta-\mu}{2})\Gamma(\frac{1+m+\beta-\mu}{2})}
\Bigg)\notag\\
    &=
(-1)^{m+1}\sum_{k=0}^{m-1}\frac{(-1+m-k)!\big(\frac{1+m+\beta-\mu}{2}\big)_{k-m}
      \big(\frac{1+m+\beta+\mu}{2}\big)_{k-m} (-z)^{k-m} }
      {\Gamma\big(\frac{1-m-\beta-\mu}{2}\big)
        \Gamma\big(\frac{1-m+\beta-\mu}{2}\big)
        k!}\notag\\
      &+         (-1)^{m+1}\sum_{k=0}^\infty\Bigg(\frac12\psi\Big(\frac{1-m-\beta-\mu}{2}\Big)
      +\frac12\psi\Big(\frac{1-m+\beta-\mu}{2}\Big)\notag\\
     & +\frac12\psi\Big(\frac{1+m-\beta-\mu}{2}\Big)
    +\frac12\psi\Big(\frac{1+m+\beta-\mu}{2}\Big)\notag\\
    &\quad
    +\frac12H_k\Big(\frac{1+m+\beta-\mu}{2}\Big)
    +\frac12H_k\Big(\frac{1+m+\beta+\mu}{2}\Big)
    -\psi(1+m+k)\notag\\
    &\quad
    +\frac12H_{k+m}\Big(\frac{1-m+\beta-\mu}{2}\Big)
    +\frac12H_{k+m}\Big(\frac{1-m+\beta+\mu}{2}\Big)
    -\psi(1+k)\Bigg)\notag
    \\
    &\times\frac{\big(\frac{1+m+\beta-\mu}{2}\big)_k
      \big(\frac{1+m+\beta+\mu}{2}\big)_k z^k }
      {\Gamma\big(\frac{1-m-\beta-\mu}{2}\big)
        \Gamma\big(\frac{1-m+\beta-\mu}{2}\big)
        (m+k)!k!}\notag\\
      &+(-1)^{m+1}\log(-z) \frac{\bF_{\alpha,\beta,\mu}(z)}{\Gamma(\frac{1-\alpha-\beta-\mu}{2})\Gamma(\frac{1-\alpha+\beta-\mu}{2})}.
\notag          \end{align}	
Finally, we simplify the expression by using a few identities:
\begin{align}
 \psi\Big(\frac{1+m+\beta-\mu}{2}\Big)
  +H_k\Big(\frac{1+m+\beta-\mu}{2}\Big)&=  \psi\Big(\frac{1+m+\beta-\mu}{2}+k\Big),\\
 \psi\Big(\frac{1-m+\beta-\mu}{2}\Big)
  +H_{k+m}\Big(\frac{1-m+\beta-\mu}{2}\Big)&
  =\psi\Big(\frac{1+m+\beta-\mu}{2}+k\Big),\\
  \psi\Big(\frac{1-m-\beta-\mu}{2}\Big)  
    +H_k\Big(\frac{1+m+\beta+\mu}{2}\Big) 
    &=\psi\Big(\frac{1-m-\beta-\mu}{2}-k\Big),
  \\
  \psi\Big(\frac{1+m-\beta-\mu}{2}\Big)
  + H_{k+m}\Big(\frac{1-m+\beta+\mu}{2}\Big)&=
  \psi\Big(\frac{1-m-\beta-\mu}{2}-k\Big).
    \end{align}
\qed

\subsection{Recurrence relations}

Recurrence relations for the hypergeometric function have a more symmetric form if we use a special normalisation, namely
			\begin{eqnarray}
			\bF^\rI_{\alpha,\beta,\mu}(z) &:=& \Gamma\Big(\frac{1+\alpha+\beta-\mu}{2}\Big)\Gamma\Big(\frac{1+\alpha-\beta+\mu}{2}\Big)\bF_{\alpha,\beta,\mu}(z)\\
			&=& \Gamma\Big(\frac{1+\alpha-\beta+\mu}{2}\Big)\sum_{k=0}^\infty\frac{\Gamma\big(\frac{1+\alpha+\beta-\mu}{2}+k\big)\big(\frac{1+\alpha+\beta+\mu}{2}\big)_k}{\Gamma(1+\alpha+k)k!}z^k.\nonumber
			\end{eqnarray}

		The function $\bF^\rI_{\alpha,\beta,\mu}$ fulfils the following recurrence relations:
		\begin{eqnarray}
			\pder_z\bF^\rI_{\alpha,\beta,\mu}(z) &=& \frac{1+\alpha+\beta+\mu}{2}\bF^\rI_{\alpha+1,\beta+1,\mu}(z) \nonumber\\
			\big(z(1-z)\pder_z + \alpha(1-z) - \beta z\big)\bF^\rI_{\alpha,\beta,\mu}(z) &=& \frac{-1+\alpha+\beta-\mu}{2}\bF^\rI_{\alpha-1,\beta-1,\mu}(z)\nonumber\\
			\big((1-z)\pder_z - \beta\big)\bF^\rI_{\alpha,\beta,\mu}(z) &=& \frac{1+\alpha-\beta-\mu}{2}\bF^\rI_{\alpha+1,\beta-1,\mu}(z)\nonumber\\
			(z\pder_z+\alpha)\bF^{\rI}_{\alpha,\beta,\mu}(z) &=&\frac{-1+\alpha-\beta+\mu}{2}\bF^\rI_{\alpha-1,\beta+1,\mu}(z)\nonumber\\
			\big(z\pder_z +\frac{1}{2}(1+\alpha+\beta+\mu)\big)\bF^\rI_{\alpha,\beta,\mu}(z) &=& \frac{1+\alpha+\beta+\mu}{2}\bF^\rI_{\alpha,\beta+1,\mu+1}(z)\nonumber\\
			\big(z\pder_z +\frac{1}{2}(1+\alpha+\beta-\mu)\big)\bF^\rI_{\alpha,\beta,\mu}(z) &=& \frac{-1+\alpha-\beta+\mu}{2}\bF^\rI_{\alpha,\beta+1,\mu-1}(z)\label{recur-Fabm}\nonumber\\
			\big(z(1-z)\pder_z - \beta + \frac{1}{2}(1+\alpha+\beta+\mu)(1-z)\big)\bF^\rI_{\alpha,\beta,\mu}(z) &=& \frac{-1+\alpha+\beta-\mu}{2}\bF^\rI_{\alpha,\beta-1,\mu+1}(z)\nonumber\\
			\big(z(1-z)\pder_z - \beta + \frac{1}{2}(1+\alpha+\beta-\mu)(1-z)\big)\bF^\rI_{\alpha,\beta,\mu}(z) &=& \frac{1+\alpha-\beta-\mu}{2}\bF^\rI_{\alpha,\beta-1,\mu-1}(z)\nonumber\\
			\big((z-1)\pder_z + \frac{1}{2}(1+\alpha+\beta+\mu)\big)\bF^\rI_{\alpha,\beta,\mu}(z) &=& \frac{1+\alpha+\beta+\mu}{2}\bF^\rI_{\alpha+1,\beta,\mu+1}(z)\nonumber\\
			\big((z-1)\pder_z + \frac{1}{2}(1+\alpha+\beta-\mu)\big)\bF^\rI_{\alpha,\beta,\mu}(z) &=& \frac{1+\alpha-\beta-\mu}{2}\bF^\rI_{\alpha+1,\beta,\mu-1}(z)\nonumber\\
			\big(z(1-z)\pder_z + \alpha - \frac{1}{2}(1+\alpha+\beta+\mu)z\big)\bF^\rI_{\alpha,\beta,\mu}(z) &=& \frac{-1+\alpha+\beta-\mu}{2}\bF^\rI_{\alpha-1,\beta,\mu+1}(z)\nonumber\\
			\big(z(1-z)\pder_z + \alpha - \frac{1}{2}(1+\alpha+\beta-\mu)z\big)\bF^\rI_{\alpha,\beta,\mu}(z) &=& \frac{-1+\alpha-\beta+\mu}{2}\bF^\rI_{\alpha-1,\beta,\mu-1}(z).\nonumber			
		\end{eqnarray}

		In order to have  recurrence relations for
 $\bD_{m,\beta,\mu}$ similar to relations for $\bF^{\rI}_{m,\beta,\mu}$, we change its normalisation:
		\begin{equation}
			\bD^{\rI}_{m,\beta,\mu}(z) := \Gamma\Big(\frac{1+m+\beta-\mu}{2}\Big)\Gamma\Big(\frac{1+m-\beta+\mu}{2}\Big)\bD_{m,\beta,\mu}(z).
		\end{equation}	
		The recurrence relations for $\log(-z)\bF^{\rI}_{m,\beta,\mu} + \bD^{\rI}_{m,\beta,\mu}$ are the same as for $\bF^{\rI}_{m,\beta,\mu}$ and they lead to the recurrence relations for
		$ \bD^{\rI}_{m,\beta,\mu}$:
		
		\begin{eqnarray}
		\pder_z\bD^\rI_{m,\beta,\mu}(z) &=&\nonumber\\ = \frac{1+m+\beta+\mu}{2}\bD^\rI_{m+1,\beta+1,\mu}(z) &-& \frac{\bF^{\rI}_{m,\beta,\mu}(z)}{z},\nonumber\\
		\big(z(1-z)\pder_z +  m (1-z) - \beta z\big)\bD^\rI_{ m ,\beta,\mu}(z) &=&\nonumber\\ = \frac{-1+ m +\beta-\mu}{2}\bD^\rI_{ m -1,\beta-1,\mu}(z)&-&(1-z)\bF^{\rI}_{m,\beta,\mu}(z),\nonumber\\
		\big((1-z)\pder_z - \beta\big)\bD^\rI_{ m ,\beta,\mu}(z) &=&\nonumber\\ =\frac{1+ m -\beta-\mu}{2}\bD^\rI_{ m +1,\beta-1,\mu}(z)&-&\frac{(1-z)}{z}\bF^{\rI}_{m,\beta,\mu}(z),\nonumber\\
		(z\pder_z+ m )\bD^{\rI}_{ m ,\beta,\mu}(z) &=&\nonumber\\ =\frac{-1+ m -\beta+\mu}{2}\bD^\rI_{ m -1,\beta+1,\mu}(z)&-&\bF^{\rI}_{m,\beta,\mu}(z),\nonumber\\
		\big(z\pder_z +\frac{1}{2}(1+ m +\beta+\mu)\big)\bD^\rI_{ m ,\beta,\mu}(z) &=&\nonumber\\ = \frac{1+ m +\beta+\mu}{2}\bD^\rI_{ m ,\beta+1,\mu+1}(z)&-&\bF^{\rI}_{m,\beta,\mu}(z),\nonumber\\
		\big(z\pder_z +\frac{1}{2}(1+ m +\beta-\mu)\big)\bD^\rI_{ m ,\beta,\mu}(z) &=&\nonumber\\ = \frac{-1+ m -\beta+\mu}{2}\bD^\rI_{ m ,\beta+1,\mu-1}(z)&-&\bF^{\rI}_{m,\beta,\mu}(z)\notag,\\
		\big(z(1-z)\pder_z - \beta + \frac{1}{2}(1+ m +\beta+\mu)(1-z)\big)\bD^\rI_{ m ,\beta,\mu}(z) &=&\nonumber\\ = \frac{-1+ m +\beta-\mu}{2}\bD^\rI_{ m ,\beta-1,\mu+1}(z)&-&(1-z)\bF^{\rI}_{m,\beta,\mu}(z),\nonumber\\
		\big(z(1-z)\pder_z - \beta + \frac{1}{2}(1+ m +\beta-\mu)(1-z)\big)\bD^\rI_{ m ,\beta,\mu}(z) &=&\nonumber\\ = \frac{1+ m -\beta-\mu}{2}\bD^\rI_{ m ,\beta-1,\mu-1}(z)&-
		&(1-z)\bF^{\rI}_{m,\beta,\mu}(z)\nonumber,\\
		\big((z-1)\pder_z + \frac{1}{2}(1+ m +\beta+\mu)\big)\bD^\rI_{ m ,\beta,\mu}(z) &=&\nonumber\\ = \frac{1+ m +\beta+\mu}{2}\bD^\rI_{ m +1,\beta,\mu+1}(z)&-&\frac{z-1}{z}\bF^{\rI}_{m,\beta,\mu}(z),\nonumber\\
		\big((z-1)\pder_z + \frac{1}{2}(1+ m +\beta-\mu)\big)\bD^\rI_{ m ,\beta,\mu}(z) &=&\nonumber\\ = \frac{1+ m -\beta-\mu}{2}\bD^\rI_{ m +1,\beta,\mu-1}(z)\nonumber&-&\frac{z-1}{z}\bF^{\rI}_{m,\beta,\mu}(z),\\
		\big(z(1-z)\pder_z +  m  - \frac{1}{2}(1+ m +\beta+\mu)z\big)\bD^\rI_{ m ,\beta,\mu}(z) &=&\nonumber\\ = \frac{-1+ m +\beta-\mu}{2}\bD^\rI_{ m -1,\beta,\mu+1}(z)&-&(1-z)\bF^{\rI}_{m,\beta,\mu}(z),\nonumber\\
		\big(z(1-z)\pder_z +  m  - \frac{1}{2}(1+ m +\beta-\mu)z\big)\bD^\rI_{ m ,\beta,\mu}(z) &=&\nonumber\\ = \frac{-1+ m -\beta+\mu}{2}\bD^\rI_{ m -1,\beta,\mu-1}(z)&-&(1-z)\bF^{\rI}_{m,\beta,\mu}(z).\nonumber			
		\end{eqnarray}
		
		They imply contiguous relations:
		\begin{eqnarray}
		m\bD^\rI_{m,\beta,\mu}(z) &=& \frac{-1+m-\beta+\mu}{2}\bD^\rI_{m-1,\beta+1,\mu}(z) - \frac{1+m+\beta+\mu}{2}z\bD^\rI_{m+1,\beta+1,\mu}(z),\nonumber\\
		m(1-z)\bD^\rI_{m,\beta,\mu}(z) &=& \frac{-1+m+\beta-\mu}{2}\bD^\rI_{m-1,\beta-1,\mu}(z) - \frac{1+m-\beta-\mu}{2}z\bD^\rI_{m+1,\beta-1,\mu}(z),\nonumber \\
		\mu\bD^\rI_{m,\beta,\mu}(z) &=& \frac{1+m+\beta+\mu}{2}\bD^\rI_{m,\beta+1,\mu+1}(z) - \frac{-1+m-\beta+\mu}{2}\bD^\rI_{m,\beta+1,\mu-1}(z),\nonumber\\
		\mu(1-z)\bD^\rI_{m,\beta,\mu}(z) &=& \frac{-1+m+\beta-\mu}{2}\bD^\rI_{m,\beta-1,\mu+1}(z) - \frac{1+m-\beta-\mu}{2}\bD^\rI_{m,\beta-1,\mu-1}(z),\nonumber\\
		\mu\bD^\rI_{m,\beta,\mu}(z) &=& \frac{1+m+\beta+\mu}{2}\bD^\rI_{m+1,\beta,\mu+1}(z) - \frac{1+m-\beta-\mu}{2}\bD^\rI_{m+1,\mu,\beta-1}(z),\nonumber\\
		\mu z\bD^\rI_{m,\beta,\mu}(z) &=& \frac{-1+m-\beta+\mu}{2}\bD^\rI_{m-1,\beta,\mu-1}(z) - \frac{-1+m+\beta-\mu}{2}\bD^\rI_{m-1,\beta,\mu+1}(z).\nonumber
		\end{eqnarray}
		
		\subsection{Kummer's table relations}

As a special case of relations from   the so-called Kummer's table, the hypergeometric function satisfies
                \begin{align}
                  \bF_{\alpha,\beta,\mu}(z)&=(1-z)^{-\beta}                  \bF_{\alpha,-\beta,\mu}(z),\label{pow1}\\
                  \bU_{\alpha,\beta,\mu}(z)&=(1-z)^{-\beta}                  \bU_{\alpha,-\beta,\mu}(z),\label{pow2}
                \end{align}
                (see e.g. \cite{D}, and also
                (\ref{also})). There is also
                an analogous identity for $\bD_m$:
              \begin{theorem}   For $m\in\bbZ$, \begin{align}
                  \bD_{m,\beta,\mu}(z)&=(1-z)^{-\beta}                  \bD_{m,-\beta,\mu}(z).
                  \end{align}\end{theorem}
              \begin{proof} We  use   (\ref{U-def})
                together with (\ref{pow1}), (\ref{pow2}). \end{proof}

\subsection{Quadratic relations}
              
                There exist also well-known ``doubling relations'' between hypergeometric functions with special parameters, which involve a quadratic transformation of the independent variable, such as
                \begin{align}
                  \Gamma(1+2\alpha)\bF_{2\alpha,\beta,-\beta}(z)
&=
                  \Big(\frac{2}{2-z}\Big)^{\frac12+\alpha+\beta}
                  \Gamma(1+\alpha)\bF_{\alpha,\beta,-\frac12}\Big(\frac{z^2}{(2-z)^2}\Big)
\label{sasa}\\
                  &=
(1-z)^{-\frac14-\frac{\alpha}{2}-\frac{\beta}{2}}                                 \Gamma(1+\alpha)\bF_{\alpha,-\frac12,-\beta}\Big(\frac{z^2}{4(z-1)}\Big),
\label{sasa2}\\
\bF_{\beta,\beta,2\alpha}(z)&=\bF_{\beta,-\frac12,\alpha}\big(4z(1-z)\big)
\label{sasa1}\\
  &=
  (1-2z)^{-\frac12-\beta-\alpha}\bF_{\beta,\alpha,-\frac12}
  \Big(\frac{4z(z-1)}{(1-2z)^2}\Big).\label{sasa5}
                \end{align}

     Indeed, we check that the functions that appear on the left and right hand sides of (\ref{sasa}) and (\ref{sasa2}) are annihilated by the hypergeometric operator $\cF_{2\alpha,\beta,-\beta}$, are analytic at $0$, and equal $1$ at $0$. Using the fact that $0$ is a regular singular point of the hypergeometric equation, we conclude that they coincide, which proves identities 
     (\ref{sasa}) and (\ref{sasa2}).
  (\ref{sasa1}) and (\ref{sasa5}) can be proven in a  similar way.

                (\ref{sasa1}) can be rewritten as
                \begin{align}
                  \bU_{2\alpha,\beta,-\beta}(z)&= \big(4(1-z)\big)^{-\frac14-\frac{\alpha}{2}-\frac{\beta}{2}}                                 \bU_{\alpha,-\frac12,-\beta}\Big(\frac{z^2}{4(z-1)}\Big).
                 \label{uu0} \end{align}
                Here is a doubling relation for the functions $\bD$:

                \begin{theorem}
                  For $m=0,1,\dots$, we have
            \begin{align}
         \bD_{2m,\beta,-\beta}(z)
                  &= \frac{m!}{2(2m)!}(1-z)^{-\frac14-\frac{m}{2}-\frac{\beta}{2}}                                 \Bigg(\bD_{m,-\frac12,-\beta}\Big(\frac{z^2}{4(z-1)}\Big)
               \notag  \\&\hspace{10ex} -\log\big(4(1-z)\big)
                  \bF_{m,-\frac12,-\beta}\Big(\frac{z^2}{4(z-1)}\Big)
               \Bigg).\label{sasa3}
\end{align}       

                \end{theorem}
		
		\begin{proof}
                  By (\ref{U-def}), we have
                  \begin{align}
&                    \bU_{2m,\beta,-\beta}(z)\label{uu1}\\
                    =&\frac{-1}{\Gamma(\frac12-m)\Gamma(\frac12-m+\beta)}
                    \Big(\bD_{2m,\beta,-\beta}(z)+
                    \log(-z)\bF_{2m,\beta,-\beta}(z)\Bigg)\\ =&\frac{(-1)^{1+m}\Gamma(1+2m)}{\sqrt{\pi}4^m\Gamma(1+m)\Gamma(\frac12-m+\beta)}
                    \Big(\bD_{2m,\beta,-\beta}(z)+
                    \log(-z)\bF_{2m,\beta,-\beta}(z)\Bigg),\label{log1}
                    \end{align}
                    where we used (\ref{gamma}).
                    Again, by (\ref{U-def}), we have
                    \begin{align}
&\big(4(1-z)\big)^{-\frac14-\frac{m}{2}-\frac{\beta}{2}}                                 \bU_{m,-\frac12,-\beta}\Big(\frac{z^2}{4(z-1)}\Big)
                    \label{uu2}  \\  =&
                      \frac{(-1)^{m+1}\big(4(1-z)\big)^{-\frac14-\frac{m}{2}-\frac{\beta}{2}}                                 }{\Gamma(\frac{1}{4}-\frac{m}{2}+\frac{\beta}{2})
                        \Gamma(\frac{3}{4}-\frac{m}{2}+\frac{\beta}{2})
                       } 
\Bigg(\bD_{m,-\frac12,-\beta}\Big(\frac{z^2}{4(z-1)}\Big)\\&\hspace{5ex}
               +\log\Big(\frac{z^2}{4(1-z)}\Big)
                  \bF_{m,-\frac12,-\beta}\Big(\frac{z^2}{4(z-1)}\Big)
                  \Bigg)\\
                   \\  =& \frac{(-1)^{m+1}(1-z)^{-\frac14-\frac{m}{2}-\frac{\beta}{2}}                                 }{2^{1+2m}\sqrt{\pi}\Gamma(\frac{1}{2}-m+\beta)
                                               } 
\Bigg(\bD_{m,-\frac12,-\beta}\Big(\frac{z^2}{4(z-1)}\Big)\\&
+\Big(2\log(-z) -
\log\big(4(1-z)\big)\Big)
                  \bF_{m,-\frac12,-\beta}\Big(\frac{z^2}{4(z-1)}\Big)
               \Bigg),\label{log2}
\end{align}       
                    where we used
                    \begin{equation}
                      \Gamma\Big(\frac{1}{4}-\frac{m}{2}+\frac{\beta}{2}\Big)
                      \Gamma\Big(\frac{3}{4}-\frac{m}{2}+\frac{\beta}{2}\Big)
                      =2^{\frac12+m-\beta}\sqrt{\pi}\Gamma\Big(\frac12-m+\beta\Big). \end{equation}
                    Now by (\ref{uu0}) we have the identity
                    (\ref{uu1})=(\ref{uu2}).
                    Then we notice that by (\ref{sasa2}) the terms in (\ref{log1}) and (\ref{log2}) involving $\log(-z)$ cancel. We obtain                    (\ref{sasa3}).
				\end{proof}

                \appendix
                \section{Some formulas}
In our paper we use various functions related to the Gamma function $\Gamma(z)$:
                           \begin{align}
\text{                the digamma function}&&   \psi(z)&:=\frac{\partial_z\Gamma(z)}{\Gamma(z)},\label{digamma}
                  \\
    \text{the shifted $k$th harmonic number}&&             H_k(z)& := \frac{1}{z} + \dots + \frac{1}{z+k-1},\label{haka}\\
    \text{the $k$th harmonic number}&&                          H_k&:=
                  \frac{1}{1} + \dots + \frac{1}{k}=H_k(1),\label{haka0}\\
                  \text{the Pochhammer symbol}&&              (z)_k&:=\frac{\Gamma(z+k)}{\Gamma(z)}\label{pochhammer1}\\
                 && &=\begin{cases}(z)(z+1)\cdots(z+k-1),& k\geq0,\\
                  \frac{1}{(z+k)(z+k+1)\cdots(z-1)},&k\leq0.\end{cases}
      \notag
\end{align}
Some of their properties are collected below:
                           \begin{align}
                             H_{k+n}(z)&=H_n(z)+H_k(z+n),\\
                             H_k(z)&=-H_k(1-z-k),\\
                  \psi(z+k)&=\psi(z)+H_k(z),\\
                  \psi(1+k)&=-\gamma+H_k,\label{psiha}\\
(z)_k&
                  =(-1)^k(1-k-z)_k.\label{pochhammer}
                           \end{align}

                \begin{align}
                  \partial_z\frac{1}{\Gamma(z)}&=-\frac{\psi(z)}{\Gamma(z)},
                  \\
                    \partial_z\frac{1}{\Gamma(z)}\Big|_{z=-n}&=(-1)^nn!,\quad n=0,1,2,...\\
                      \partial_z(z)_n&=H_n(z)(z)_n.
                      \end{align}



\begin{thebibliography}{aaa}


\bibitem{Ar} W.J.Archibald, \emph{The complete solution of the differential equation for the confluent hypergeometric function.} Phil. Mag. VII, 26 (1938) 415-419

\bibitem{AW} J.R.Airey, H.A.Webb, \emph{The practical importance of the confluent hypergeometric function.} Phil. Mag. 36 (1918) 129-141
                  
			\bibitem{B}
			H. Buchholz, \emph{The confluent hypergeometric function}, Springer Tracts in Natural Philosophy, 1969.
			
			\bibitem{D}
			J. Derezi\'nski,
			\emph{Hypergeometric Type Functions and Their Symmetries},
			Ann. Henri Poincar\'e {\bf 15} (2014), 1569--1653.
			
			\bibitem{DM}
			J. Derezi\'nski and P. Majewski,
			\emph{From conformal group to symmetries of hypergeometric type equations},
                        SIGMA 12 (2016), 108, 69 pages.

\bibitem{H} S. Hollands, \emph{Renormalized Quantum Yang-Mills Fields in Curved Spacetime},
Rev. Math. Phys. 20 (2008) 1033-1172

                        
			\bibitem{L1}
			L.J. Slater, \emph{Confluent hypergeometric functions},
			Cambridge University Press, 1960.
			
			\bibitem{L2}
			L.J. Slater, \emph{Generalized hypergeometric functions},
			Cambridge University Press, 1966.
			
			\bibitem{MO}
			  W. Magnus and F. Oberhettinger, \emph{Formulas and Theorems for the Functions of Mathematical Physics}, English translation, Chelsea Publishing Company, 1954.







                       
			\bibitem{NIST}			National Institute of Science and Technology, \emph{Digital Library of Mathematical Functions}, \url{http://dlmf.nist.gov}, (Access: 7$^{\mathrm{th}}$ April 2017).



                          
	

                        \bibitem{W}
			G. N. Watson, \emph{A treatise on the theory of Bessel functions}, Cambridge University Press, 1922.



                          
			\bibitem{W-Wh}
			G. N. Watson and E. T. Whittaker, \emph{A course in modern analysis. An introduction to the general theory of infinite processes and of analytic functions; with an account of the principal transcendental functions}, Cambridge University Press, 1927.

		\end{thebibliography}
\end{document}